\documentclass[12pt,final]{article}
\usepackage{hyperref}
\usepackage{gt-style}
\usepackage{graphicx}
\usepackage{svg}

\newcommand\JournalTitle[1]{#1}

\usepackage{endnotes}

\usepackage{bm}

\title{Modeling the diffusion of a fluid in a strained solid: a comparison between different formats}
\author{Giuseppe Tomassetti\footnote{Dipartimento di Ingegneria, Universit\`a degli Studi Roma Tre. Via Vito Volterra 62, 00146, Roma, Italy. Email:\texttt{giuseppe.tomassetti@uniroma3.it}, \texttt{ORCID:}\url{https://orcid.org/0000-0001-8801-7461}}}

\begin{document}
\begin{sloppypar}

\maketitle
\begin{abstract}
We revise the format proposed by Bowen in [R.~M. Bowen.
\newblock {\em\JournalTitle{Int. J. Eng. Sci.}},
  18(9):1129--1148, 1980.] to describe incompressible porous media by use of the theory of mixtures. We then show that this format is equivalent the those proposed more recently to model the diffusion of a fluid in a strained solid.
  \end{abstract}
\section{Introduction.}
Mechanical theories describing polymer gels, and, more generally, a strained solid through which diffusion of a fluid takes place, have been put together using two different approaches. The early papers \cite{Bowen1980,ShiRajagopalWineman1981} and the more recent work 
\cite{CaldeCLZ2010JoCaTN} model the gel as the superposition of two interacting continuous bodies: an incompressible solid, representing the polymeric network, and an incompressible fluid. The equations governing the evolution of these continua are obtained through an adaptation of Truesdell's theory of interacting continua \cite{Truesdell1962JCP}. Another series of papers \cite{BaekSrinivasa2004,HongZZS2008JMPS,ChestA2010JMPS,DudaSF2010JMPS}, adapt the original ideas of Gibbs \cite{Gibbs1878} and Biot \cite{Biot1972IUMJ,Biot1941Joap} to model a polymeric gel as a single continuum whereby diffusion of a chemical species takes place, driven by the gradient of its chemical potential. Within this theory, constituive equations are obtained from a dissipation principle that takes explicitly into account the energetic flux associated to the motion of the fluid relative to the solid. 

Given this state of matters provides a motivation to investigatate the relationships between the the two approaches. The contribution we give in this paper is twofold: we offer a revised exposition of the ideas in \cite{Bowen1980} in the special case when when the intercting continua are incompressible; we then show that the isothermal version of dissipation principle underlying the theory of gels within the setting of the \emph{\`a la Biot} approach can be recovered, by making the appropriate identifications, from the corresponding principle in the setting of the theory of interacting continua.

\section{The compound continuum}
\subsection{Kinematics} 
\paragraph{Index convention.}
In the foregoing developments, the free index $\alpha$ runs between $1$ and the number $N$ of constituents, while the free index $\beta$ runs between $2$ and $N$:
\[
\alpha\in \{1,\ldots,N\},\qquad\beta\in\{2,\ldots,N\}.
\]
The same convention is used for summations. Thus,
\[
\sum_\alpha\equiv\sum_{\alpha=1}^N,\qquad\text{and}\qquad \sum_\beta\equiv\sum_{\beta=2}^N.
\]
The subscript $1$ shall be replaced by $s$ (standing for ``solid'') and, when there is only one fluid component, the subscript $2$ shall be replaced by $f$ (standing for ``fluid''). Accordingly:
\begin{itemize}
\item if $N=2$ then $s\equiv 1$ and $f\equiv 2$.
\end{itemize}
\normalsize

\subsubsection{The description of the motion of a compound continuum.}
Generically speaking a \emph{compound continuum} is a collection of $N$ continuum bodies, which we refer to as \emph{constituents}, which can permeate each other by occupying the same region of space at the same time. The motion of each body is described in the same fashion as in standard continuum mechanics: one introduces a reference configurations $\mathcal B_\alpha$, where $\alpha=1,\ldots,N$ is the index that labels the constituent; then the motion of the $\alpha$--th constituent of the mixture is specified by a function
\[
\mb x=\bm\chi_\alpha(\bm X,t)
\]
which associates to each material particle $\mb X$ in $\mathcal B_\alpha$ the position $\mb x$ occupied by that particle at time $t$. The functions $\bm\chi_\alpha(\mb X,t)$ are assumed to be smooth. In particular, the deformation gradients 
\begin{equation}\label{eq:21}
\operatorname{Grad}\bm\chi_\alpha:=\PART{\bm\chi_\alpha}{\mb X},
\end{equation}
and the velocities
\begin{equation}\label{eq:22}
\dot{\bm\chi}_\alpha:=\PART{\bm\chi_\alpha}{t},
\end{equation}
are well defined. For $t$ fixed, the function $\bm\chi_\alpha(\cdot,t)$ is referred to as the \emph{configuration} of the  $\alpha$--th constituent at time $t$. \medskip

\subsubsection{Eulerian and Lagrangean descriptions of a field.}
Definitions \eqref{eq:21} and \eqref{eq:22} express the deformation gradient and the velocity as functions of the referential label $\mb X$. We say that these definitions provide a Lagrangean representation of these fields. 

Now, mixture theory inherits from continuum mechanics the requirement that each configuration be a one-to-one mapping, thus forbidding interpenetration between parts of the same body. Because of this fact, given a Lagrangean field we can provide a Eulerian representation where the independent variable $\mb X$ is replaced by $\mb x$. For example, we define the spatial (Eulerian) velocity field of constituent $\alpha$ as:\footnote{Here $\bm\chi^{-1}(\cdot,t)$ denotes the inverse of the configuration $\mb\chi_\alpha(\cdot,t)$ at time $t$.} 
\[
\mb v_\alpha(\mb x,t)=\dot{\bm\chi}_\alpha(\bm\chi_\alpha^{-1}(\mb x,t),t).
\]
Of course, the reciprocal operation starting from a spatial field $\varphi(\bm x,t)$ can be carried out to define its Lagrangean (referential) representation
\[
\Phi(\mb X,t)=\varphi(\bm\chi(\mb X,t),t).
\]
In particular, given a spatial field $\varphi_\alpha(\mb x,t)$ related to component $\alpha$, the \emph{material time derivative} of $\varphi$ \emph{following the motion of the constituent $\alpha$} is denoted by a grave accent:
\begin{equation}\label{eq:15}
  \grave\varphi_\alpha:=\PART{\Phi_\alpha}{t}=\PART{\varphi_\alpha}{t}+\gradop\varphi_\alpha\cdot\mb v_\alpha.
\end{equation}
When computing the time derivative of field that has already a subscript $\alpha$, we shall omit the double specification of the index $\alpha$ on top of the dot. If the field refers to a component $\beta$ different from $\alpha$, or to no component at all (in which case $\beta$ is to be omitted) we write
\begin{equation}\label{eq:23}
  \stackrel{\alpha}{\grave\varphi}_\beta:=\PART{\Phi_\beta}{t}=\PART{\varphi_\beta}{t}+\gradop\varphi_\beta\cdot\mb v_\alpha.
\end{equation}

\paragraph{Convention.}
Although $\varphi$ and $\Phi$ are different functions, they represent the same physical quantity. To avoid the proliferation of symbols, we shall avoid when possible using different symbols for the Lagrangian and Eulerian descriptions of a field. Which description is to be understood at a particular point of our development should be clear from the context.

\subsection{Partial stresses and internal forces}
A difference between mixture theory and conventional continuum theories is that different constituents can engage the same region of space. In fact, the images of two configurations $\bm\chi_\alpha(\cdot,t)$ and $\bm\chi_\beta(\cdot,t)$ ($\alpha\neq \beta$) need not be disjoint. In other words, two distinct bodies may interpenetrate each other and interact. The token used in mixture theory to model the interactions between superposed constituent is an internal body force field that accompanies contact forces and the external body forces. The equilibrium equation that affirm balance between stress of a constituent, the external and the internal forces acting on a constituent may be derived through the principle of virtual powers. Additional properties shall be derived by recourse to invariance principles.

\subsubsection{Internal and external powers.} Given a spatial region $\Omega$, the internal mechanical power expended within $\Omega$ is a linear functional of the individual velocities of the mixture costituents:
\begin{equation}\label{eq:61}
  \mathscr W_{\rm int}(\Omega)[\mb v_1,\ldots,\mb v_n]:=\sum_{\alpha}\int_\Omega \left(\mb T_\alpha\cdot\gradop\mb v_\alpha+\mb f_\alpha\cdot\mb v_\alpha\right).
\end{equation}
The external power expended on the part occupying the spatial region $\Omega$ is
\begin{equation}
  \mathscr W_{\rm ext}(\Omega)[\mb v_1,\ldots,\mb v_n]:=\sum_{\alpha}\left(\int_{\partial\Omega} \mb t_\alpha\cdot\mb v_\alpha+\int_\Omega \mb b_\alpha\cdot\mb v_\alpha\right).
\end{equation}

\subsubsection{Principle of virtual powers and balance equations.}

The principle of virtual powers for every region $\Omega$ yields the pointwise form of the equilibrium  equations
\begin{equation}\label{eq:35}
-\divop\mb T_\alpha+\mb f_\alpha=\mb b_\alpha,
\end{equation}
as well as the \emph{Cauchy's representation}:
\begin{equation}\label{eq:52}
\mb t_\alpha(\mb x,\mb n)=\mb  T_\alpha(\mb x)\mb n
\end{equation}
of the contact force $\mb t_\alpha$ acting at point $\mb x$ on the boundary of a region $\Omega$ whose outward unit normal is $\mb n$. 

\subsubsection{Invariance of the internal power.}
The invariance of the internal power over superposed rigid velocities yields that the internal forces add up to null:
\begin{equation}\label{eq:2}
\sum_\alpha\mb f_\alpha=\mb 0,
\end{equation}
and that the \emph{total stress}:
\[
\mb T:=\sum_\alpha\mb T_\alpha\in\Sym.
\]
is a symmetric tensor:
\[
\mb T\in\Sym.
\]
The fields $\mathbf T_\alpha$ are called \emph{partial stresses}. Although they need not be symmetric, their sum must be symmetric.

\subsection{Further  bits of kinematics}
\subsubsection{The bulk specific volume.}
Since we leave inertial effects out of the picture, we shall not introduce the concept of mass and mass density. Instead, we shall make use of the notion of bulk specific volumes
\[
\phi_\alpha=\frac 1 {\det\mb F_\alpha},
\]
which obey the same laws of conservation as for mass in continuum mechanics:
\begin{lemma}\label{lem:222}
The bulk specific volumes obey:
\begin{subequations}
\begin{equation}
\grave\phi_\alpha+\phi_\alpha\divop{\mb v_\alpha}=0,
\end{equation}
and, equivalently,
\begin{equation}\label{eq:11}
\PART{\phi_\alpha}{t}+\divop(\phi_\alpha\mb v_\alpha)=0.
\end{equation}
\end{subequations}
\end{lemma}
The interpretation of the bulk specific volume is the following: consider a part $\mathcal P_\alpha$ of the body $\mathcal B_\alpha$. At a given time $t_0$, the material particles within the part $\mathcal P_\alpha$ will engage the spatial region 
\[
\Omega_\alpha=\bm\chi_\alpha(\mathcal P_\alpha,t_0).
\]
If the diameter of  $\mathcal P_\alpha$ is small compared to the characteristic scale of oscillation of the fields of interest (typically, this scale may be identified with the norm $|\mathbf F_\alpha|$ of the deformation gradient), then the volume of the spatial region $\Omega_\alpha$ is approximately:\footnote{Let consider any extensive physical quantity that measures the amount of material of constituent (mass, number of molecules, etc.). Let $c_\alpha$ and $C_\alpha$ be the densities of this quantity in the current and reference configuration. Under the assumption that there are no chemical reactions that convert one constituent into another, the amount of material contained in a part $\mathcal P_\alpha$ is equal to that contained in the region $\bm\chi_\alpha(\mathcal P_\alpha,t)$ occupied by $\mathcal P_\alpha$ at time $t$:
\[
\int_{\mathcal P_\alpha}C_\alpha(\mb X)=\int_{\bm\chi_\alpha(\mathcal P_\alpha,t)}c_\alpha(\mb x,t).
\]
This yields
\[
\phi_\alpha=\frac{c_\alpha}{C_\alpha}.
\]
In particular, if $C_\alpha$ does not depend on $X$ in the reference configuration, then $\phi_\alpha$ is proportional to the density $c_\alpha$. }
\[
\operatorname{vol}\Omega_\alpha\simeq \phi_\alpha \operatorname{vol}\mathcal P_\alpha.
\]
Thus, $\phi_\alpha$ accounts for how much of spatial volume is engaged by a unit volume in the reference configuration.

\subsubsection{The true specific volume.}
When referring to the region $\Omega_\alpha$ as being ``engaged'' by the part $\mathcal P_\alpha$ we purportedly do not use the verbal adjective ``occupied''. In fact, there is a class of physical systems, called \emph{multicomponent mixtures}, that can be modeled by the theory of interacting continua (for example, a mixture of tiny air bubbles in water), but whose constituents are not truly  superposed. For these systems, an observation performed at an intermediate scale, smaller than that at which the theory gives good results (but still larger than that at which the continuum hypothesis breaks down) reveals that the constituents are not at all superposed: they occupy disjoint regions within $\mathcal S_\alpha$. Although these regions are so finely mixed that they cannot be distinguished at the macroscopic scale, within each region one can find only one component, whose true specific volume need not be the same as its bulk specific volume.

\subsubsection{The incompressibility constraint.}
For a single constituent, the incompressibility constraint consists in the requirement that the true specific volume be equal to 1. By drawing a cartoon one immediately realizes that:
\begin{itemize}
\item for a multicomponent mixture whose constituents are incompressible, the bulk specific volume $\phi_\alpha$ is the  fraction of the volume occupied by the $\alpha$--th constituent in a mesoscopic region of space.
\end{itemize}
Since volume fractions add up to 1, for an incompressible multicomponent mixture the following incompressibility constraint holds: 
\begin{equation}\label{eq:9}
\sum_\alpha\phi_\alpha=1.
\end{equation}
\begin{lemma}[Velocity constraint imposed by incompressibility]
Every velocity field consistent with \eqref{eq:9} must satisfy
\begin{equation}\label{eq:7}
\sum_\alpha\left(\phi_\alpha\divop{\mb v_\alpha}+\gradop\phi_\alpha\cdot\mb v_\alpha\right)=0.  
\end{equation}
\end{lemma}

\begin{remark}
\rm One would be tempted to argue from \eqref{eq:9} that $\sum_\alpha\grave\phi_\alpha=0$. Note however that the partial derivative in this sum are performed following the individual constituents. 
\end{remark}

\subsubsection{The mean velocity.}
We define the material derivative of $\varphi$ following the motion of the mixture as
\begin{equation}
\dot\varphi:=\partial_t\varphi+\gradop\varphi\cdot\mb v,
\end{equation}
where  
\begin{equation}\label{eq:8}
  \mb v=\sum_\alpha\phi_\alpha\mb v_\alpha.
\end{equation}
is the average velocity of the mixture. 
\begin{lemma}
\label{lem:1}
Assume that $\varphi$ is continuously differentiable with respect to $\mb x$ and $t$. Then
\[
\dot\varphi=\sum_\alpha \phi_\alpha\stackrel{\alpha}{\grave\varphi},
\]
\end{lemma}

\begin{remark}[The velocity of the mixture]
{\rm
The choice of the weights in \eqref{eq:8} appears arbitrary. As a possible selection criterion for a definition of the velocity $\mb v$ of the mixture would be to asking that the power expended by the total stress on the velocity be equal to the sum of the individual powers expended within the constitutents:
\[
\mathbf T\cdot\gradop\mathbf v=\sum_{\alpha}\left(\mathbf T_\alpha\cdot\gradop\mb v_\alpha+\mb f_\alpha\cdot\mb v_\alpha\right).
\]
}
\end{remark}

\subsection{Three formulations of the dissipation principle for a system of superposed continua}
\subsubsection{Partwise and pointwise formulation.}
Given a region $\Omega$ let $\Omega_\alpha(t)$ be a time dependent region convecting with the component $\alpha$
\[
\Omega_\alpha(t)=\bm\chi_\alpha(\mathcal P_\alpha,t),
\]
such that
\[
\Omega_\alpha(t_0)=\Omega.
\]
The dissipation principle dictates that there exists  referential state fields $\Psi_\alpha$ satisfying the inequality
\[
\sum_\alpha\left.\Dt\right|_{t=t_0}\int_{\mc P_\alpha}\Psi_\alpha-\int_\Omega\lambda\sum_\alpha\left(\phi_\alpha\divop{\mb v_\alpha}+\gradop\phi_\alpha\cdot\mb v_\alpha\right)\le \mathscr W_{\rm int}(\Omega)[\mb v_1,\dots,\mb v_n]
\]
for every realizable process and for every region $\Omega$. Here $\lambda$ is the pressure, the Lagrange multiplier associated to the incompressibility constraint \eqref{eq:7}.
\begin{proposition}
If the dissipation inequality holds true for every region $\Omega$ and for every process, then the following inequality holds pointwise:\footnote{Here $\Psi_\alpha$ stands for the Eulerian representation of the specific free energy.}
\[
\sum_\alpha\left(-\phi_\alpha\grave\Psi_\alpha+(\mb T_\alpha+\phi_\alpha\lambda\mb I)\cdot\gradop\mb v_\alpha+(\mb f_\alpha+\lambda\gradop\phi_\alpha)\cdot\mb v_\alpha\right)\ge 0. 
\]
\end{proposition}
\subsubsection{Pointwise formulation using spatial densities: the chemical--potential tensor.}
On introducing the \emph{spatial} free energy densities
\[
\psi_\alpha=\phi_\alpha\Psi_\alpha,
\]
and owing to the identities  ${\grave\psi}_\alpha={\grave\phi}_\alpha\Psi_\alpha+\phi_\alpha{\grave\Psi}_{\alpha}$ and ${\grave\phi}_\alpha=-\phi_\alpha\divop\mb v_\alpha$, we can rewrite the pointwise dissipation inequality as 
\begin{equation}\label{eq:12}
  -\sum_\alpha\left({\grave\psi}_\alpha+\phi_\alpha(\mb K_\alpha-\lambda\mb I)\cdot\gradop\mb v_\alpha\right)+\sum_{\beta}(\mb f_\beta+\lambda\gradop\phi_\alpha)\cdot\mb v_\alpha\ge 0,
\end{equation}
where
\begin{equation}\label{eq:34}
\mb K_\alpha:=\Psi_\alpha\mb I-\phi_\alpha^{-1}\mb T_\alpha
\end{equation}
is called \emph{chemical potential tensor} of the $\alpha$--th constituent.\footnote{See Eq. 2.42 of the paper \cite{Bowen1980} by Bowen.}

\subsubsection{Formulation using spatial control volumes: the total free energy and the boundary flux.}
We introduce the total free energy density:
\[
\psi=\sum_{\alpha} \psi_\alpha.
\]
\begin{proposition}
The dissipation inequality is equivalent to
\begin{equation}\label{eq:27}
\Dt\int_\Omega \psi+\int_{\partial\Omega}\psi_\alpha\mb v_\alpha\cdot\mb n\le \mathscr W_{\rm int}(\Omega)[\mb v_1,\dots,\mb v_n].
\end{equation}
\end{proposition}
\noindent The second term on the left--hand side is interpreted as the boundary influx of free energy.

\subsection{Constitutive equations}
\subsubsection{Partial free energies.} We assume that the referential free energy densitiy of each component $\alpha$ depends on the deformation gradient of the solid component and on the volume fractions of the fluid components:
\begin{equation}\label{eq:21}
\Psi_\alpha=\widehat\Psi_{\alpha}(\mathbf F,\phi_2,\ldots,\phi_N)\equiv \widehat\Psi_{\alpha}(\mathbf F,\phi_\beta),\qquad 
\end{equation}
Here and in the following we omit specifying the subscript of the deformation gradient of the solid component:\footnote{This will not be source of ambiguities, since it is only the deformation gradient of the solid component that is relevant to our development.} 
\[
\mb F\equiv\mb F_s.
\]
As apparent from \eqref{eq:21}, the free energy density of the $\alpha$--th constituent depends not only on the configuration $\bm\chi_\alpha(\cdot,t)$ of that constituent, but also on the configurations of the other constituents. The corresponding spatial free energy densities obey the following constitutive precriptions:\footnote{Here the symbol  $\phi_\beta$ in the argument of $\widehat\psi_s$ and $\widehat\Psi_s$ stands for the ordered list $(\phi_\beta)_{\beta\ge 2}$ of the volume fractions of the fluid components.}
\[
\widehat{\psi}_s(\mathbf F,\phi_\beta)=\det{\mb F}^{-1}\widehat\Psi_s(\mathbf F,\phi_\beta),\qquad 
\widehat\phi_\beta(\mathbf F,\phi_\beta)=\phi_\beta\widehat\Psi_\beta(\mathbf F,\phi_\beta).
\]
The total free energy of the mixture per unit spatial volume is the sum of the individual energies
\begin{equation}
  \label{eq:1}
\psi=\sum_{\alpha} \widehat\psi_\alpha(\mathbf F,\phi_\beta)=:\widehat\psi(\mb F,\phi_\beta).  
\end{equation}
It is important to notice that the free energy of the $\alpha$--th component of the mixture depends not only on $\mb F_\alpha$, but also on the deformation gradients of all other components. \medskip
\begin{proposition}\label{prop:1}
The dissipation inequality \eqref{eq:12} is satisfied for all admissible processes if and only if
\begin{equation}\label{eq:30}
  -\sum_{\alpha}\phi_\alpha\mb K_\alpha^{(d)}\cdot\gradop\mb v_\alpha+\sum_{\alpha}\mb f_\alpha^{(d)}\cdot\mb v_\alpha\ge 0,
\end{equation}
where
\begin{subequations}\label{eq:33}
\begin{align}
  &\mb K_s^{(d)}=\mb K_s-\lambda\mb I+\phi_s^{-1}\frac{\partial\psi}{\partial\mb F}\mb F^{T},\label{eq:33a}\\
  &\mb K_\beta^{(d)}=\mb K_\beta-\lambda\mb I-\frac{\partial\psi}{\partial\phi_\beta}\mb I,\qquad\text{for}\qquad\beta\ge 2,
\end{align}
\end{subequations}
and
\begin{subequations}\label{eq:28}
\begin{align}
&\mb f_s^{(d)}=\mb f_s+\lambda\gradop\phi_s-\gradop\psi_s+\PART{\psi}{\mb F}:\gradop\mb F,\\
&\mb f_\beta^{(d)}=\mb f_\beta+\lambda\gradop\phi_\beta-\gradop\psi_\beta+\PART{\psi}{\phi_\beta}\gradop\phi_\beta,\quad\beta\ge 2.
\end{align}
\end{subequations}
where $\PART{\Psi}{\mb F}:\gradop\mb F=\PART{\psi}{F_{ij}}{F_{ij,k}}\mb e_k$.
\end{proposition}
\begin{proof}[Proof (sketch)]
A routine computation gives based on \eqref{eq:15}, \eqref{eq:11}, and on the identity
\[
\gradop\mb v_s=\dot{\mb F}\mb F^{-1},
\] 
yields
\[
\begin{aligned}
\sum_\alpha\grave\psi_\alpha=&\sum_\alpha\frac{\partial\psi_\alpha}{\partial t}+\sum_\alpha \gradop\psi_\alpha\cdot\mb v_\alpha\\
=&\frac{\partial\psi}{\partial t}+\sum_\alpha \gradop\psi_\alpha\cdot\mb v_\alpha\\
=&\PART\psi{\mb F}\mb F^T\cdot\gradop\mb v_s-\sum_\beta\phi_\beta\PART{\psi}{\phi_\beta}\divop\mb v_\beta\\
&+\gradop\psi_s\cdot\mb v_s-\PART\psi{\mb F}\cdot\gradop\mb F[\mb v_s]\\
&+\sum_\beta(\gradop\psi_\beta-\PART{\psi}{\phi_\beta}\gradop\phi_\beta)\cdot\mb v_\beta.
\end{aligned}
\]
Substitution into \eqref{eq:12} yields the thesis. 
\end{proof}

\begin{remark}
{\rm
We notice that since $\psi=\sum_\alpha\psi_\alpha$ and $\sum_\alpha\phi_\alpha=1$, the energetic parts of the internal forces, defined by
\begin{equation}
\begin{aligned}
\mb f_s^{(e)}&=-\lambda\gradop\phi_s+\gradop\psi_s-\PART{\psi}{\mb F}:\gradop\mb F,\qquad \text{and}\\
\mb f_\beta^{(e)}&=-\lambda\gradop\phi_\beta+\gradop\psi_\beta-\PART{\psi}{\phi_\beta}\gradop\phi_\beta,\quad\beta\ge 2,
\end{aligned}
\end{equation}
add up to null:
\begin{equation}
  \sum_{\alpha}\mb f_\alpha^{(e)}=\mb 0.
\end{equation}
}
\end{remark}

\begin{remark}
The splitting \eqref{eq:33} can also be written as $\mb K_\alpha=\mb K_\alpha^{(e)}+\mb K_\alpha^{(d)}$, with
\begin{subequations}\label{eq:33}
\begin{align}
  &\mb K_s^{(e)}=-\phi_s^{-1}\frac{\partial\psi}{\partial\mb F}\mb F^{T}+\lambda\mb I,\label{eq:33a}\\
  &\mb K_\beta^{(e)}=\frac{\partial\psi}{\partial\phi_\beta}\mb I+\lambda\mb I,\qquad\text{for}\qquad\beta\ge 2,
\end{align}
\end{subequations}
the energetic parts of the chemical potential tensors.
\end{remark}

\begin{remark}[Alternative expressions for the stresses]\label{rem:alt}
{\rm
We note that, since $\phi_\alpha=1/\det\mb F_\alpha$, on letting $\det\mb F_\alpha=J_\alpha$, we can write
\begin{equation}
\frac{\partial\phi_\alpha}{\partial\mb F_\alpha}=\frac{\partial J_\alpha^{-1}}{\partial\mb F_\alpha}=-J_\alpha^{-1}\mb F_\alpha^{-T}=-\phi_\alpha\mb F_\alpha^{-T},
\end{equation}
hence
\begin{equation}\label{eq:48}
  \PART{(\phi_\alpha\Psi_\alpha)}{\mb F_\alpha}\mb F_\alpha^T=\phi_\alpha\PART{\Psi_\alpha}{\mb F_\alpha}\mb F^T_\alpha-\phi_\alpha\Psi_\alpha\mb I.
\end{equation}
We also notice that if $\Psi_\beta$ depends on $\mb F_\alpha$ only through the volume fraction $\phi_\alpha$, then 
\begin{equation}
  \frac{\partial(\phi_\beta\Psi_\beta)}{\partial\bm F_\alpha}\mb F_\alpha^T=-\phi_\beta\phi_\alpha\frac{\partial\Psi_\beta}{\partial\phi_\alpha}\bm I.
\end{equation}
From these observations, \eqref{eq:36} can be given the form
\begin{subequations}\label{eq:53}
\begin{align}
&\phi_s^{-1}\mb T_s=\PART{\Psi_s}{\mb F}\mb F^T-\left(\lambda+\sum_{\beta\ge 2}\phi_\beta\frac{\partial\Psi_s}{\partial\phi_\beta}\right)\mb I-\mb K_s^{(d)},\label{eq:36a}\\  
&\phi_\beta^{-1}\mb T_\beta=-\left(\lambda+\sum_\alpha\phi_\alpha\PART{\Psi_\alpha}{\phi_\beta}\right)\mb I-\mb K_\beta^{(d)},\qquad \beta\ge 2.
\end{align}
\end{subequations}
In some applications, it is convenient to express the dependence of the free energy of the solid component in the form
\begin{equation}
  \widehat\Psi_\alpha(\mb F,\phi_\beta)=\overline\Psi_\alpha(\mb F,\phi_s,\phi_\beta).
\end{equation}
In these cases, the expression \eqref{eq:53} are replaced by
\begin{align}
&\phi_s^{-1}\mb T_s=\PART{\Psi_s}{\mb F}\mb F^T-\left(\lambda+\sum_{\alpha}\phi_\alpha\frac{\partial\overline\Psi_\alpha}{\partial\phi_s}\right)\mb I-\mb K_s^{(d)},\\
&\phi_\beta^{-1}\mb T_\beta=-\left(\lambda+\sum_\alpha\phi_\alpha\PART{\overline\Psi_\alpha}{\phi_\beta}\right)\mb I-\mb K_\beta^{(d)},\qquad \beta\ge 2.
\end{align}
The scalar field
\begin{equation}
  \pi_\alpha
=\sum_{\alpha}\phi_\alpha\frac{\partial\overline\Psi_\alpha}{\partial\phi_s}
\end{equation}
is the \emph{osmotic pressure of the $\alpha$-th component}.
}
\end{remark}

\subsubsection{Reduced balance equations.}
From definition \eqref{eq:34} of the chemical potential tensor we get the expression $\bm T_\alpha=-\phi_\alpha\bm K_\alpha+\psi_\alpha\mb I$ of the partial stress in terms of chemical potential tensor and partial free energy. Substitution of this expression into \eqref{eq:33} yields a set of \emph{reduced constitutive prescriptions} for the stresses:
\begin{subequations}\label{eq:36}
\begin{align}
&\mb T_s=\PART{\psi}{\mb F}\mb F^T+\phi_s(\Psi_s-\lambda)\mb I-\phi_s\mb K_s^{(d)},\label{eq:36a}\\  
&\mb T_\beta=-\phi_\beta\PART{\psi}{\phi_\beta}\mb I+\phi_\beta\left(\Psi_\beta-\lambda\right)\mb I-\phi_\beta\mb K_\beta^{(d)}.
\end{align}
\end{subequations}
These prescriptions, when substituted along with \eqref{eq:28} into the equilibrium equations \eqref{eq:35} yield
\begin{subequations}\label{eq:37}
  \begin{align}
    &\divop\PART{\psi}{\mb F}\mb F^T-\phi_s\gradop\lambda+\PART{\psi}{\mb F}:\gradop\mb F+\mb b_s=\mb f_s^{(d)}+\divop\left(\phi_s\mb K_s^{(d)}\right)\label{eq:37a}\\
    &-\phi_\beta\gradop\left(\PART{\psi}{\phi_\beta}+\lambda\right)+\PART{\psi}{\phi_\beta}\gradop\phi_\beta+\mb b_\beta=\mb f_\beta^{(d)}+\divop\left(\phi_\beta\mb K_\beta^{(d)}\right).\label{eq:37b}
  \end{align}
\end{subequations}

\subsubsection{Alternative format for the reduced balance equations.} 
From the condition $\sum_\alpha\mb f_\alpha=\mb 0$ (which we recall is a consequence of frame indifference of the internal power) and from \eqref{eq:28} we find $\sum_\alpha \mb f_\alpha^{(d)}=\mb 0$. It is convenient to write both the equilibrium equations and the reduced dissipation inequality in a form that does not involve $\mb f_s$ explicitly. As to the equilibrium equations, this is done by adding the equations \eqref{eq:37b} to \eqref{eq:37a} to get:\footnote{Note that $\PART{\psi}{\mb F}:\gradop\mb F+\sum_\beta\PART{\psi}{\phi_\beta}\gradop\phi_\beta=\gradop\psi=\divop(\psi\mb I)$.} 
\begin{subequations}\label{eq:37+}
  \begin{align}
    &\divop\left(\PART{\psi}{\mb F}\mb F^T+\psi\mb I-\sum_{\beta}\phi_\beta\frac{\partial\psi}{\partial\phi_\beta}\right)-\gradop\lambda+\mb b=\divop\left(\sum_\alpha\phi_\alpha\mb K_\alpha^{(d)}\right)\label{eq:37a+}\\
    &-\phi_\beta\gradop\left(\PART{\psi}{\phi_\beta}+\lambda\right)+\mb b_\beta=\mb f_\beta^{(d)}+\divop\left(\phi_\beta\mb K_\beta^{(d)}\right).
  \end{align}
\end{subequations}
The reduced dissipation inequality, on the other hand, is equivalent to
\begin{equation}\label{eq:29}
  -\sum_{\alpha}\phi_\alpha\mb K_\alpha^{(d)}\cdot\gradop\mb v_\alpha+\sum_{\beta}\mb f_\beta^{(d)}\cdot(\mb v_\beta-\mb v_s)\ge 0.
\end{equation}
\begin{remark}
{\rm
The weighted sum $\psi=\sum_\alpha\phi_\alpha\Psi_\alpha$ represents the free energy per unit volume of the mixture. It is immediately seen that the total Cauchy stress $\mb T=\sum_{\alpha}\mb T_\alpha$ depends only on $\psi$. Indeed, from \eqref{eq:36} we have:\footnote{See also Eq. 3.22 of Bowen's paper \cite{Bowen1980}, where the total stress is called the \emph{inner part of the stress} and is denoted by $\mb T_I$.}
\begin{equation}
  \mb T:=\sum_\alpha\mb T_\alpha=\frac{\partial\psi}{\partial\mb F}\mb F^T+(\psi-\lambda)\mb I-\sum_{\beta}\phi_\beta\frac{\partial\psi}{\partial\phi_\beta}-\sum_\alpha\phi_\alpha\mb K^{(d)}_\alpha.
\end{equation}
If we add the partial balances \eqref{eq:35} and we use the fact that the internal forces add up to null, we obtain
\begin{equation}\label{eq:47}
  -\operatorname{div}\mb T=\mb b,
\end{equation}
which coincides with \eqref{eq:37+}. In the case of a body made of a single incompressible constituent, the above equation reduces to
\begin{equation}\label{eq:48}
  \mb T=\frac{\partial\psi}{\partial\mb F}\mb F^T+(\psi-\lambda)\mb I.
\end{equation}
 In this case the referential and spatial energies coincide, because the Jacobian of the deformation map is equal to 1, and so we recover the standard constitutive equation.
}
\end{remark}

\subsubsection{Evolution equations.}
It is immediate from Proposition \ref{prop:1} that constitutive equations consistent with the dissipation inequality \eqref{eq:29} are
\begin{subequations}\label{eq:16}
\begin{align}
&\mb K_\alpha^{(d)}=\mb 0,\label{eq:16a}\\
&\mb f_\beta^{(d)}=\phi_\beta k_\beta(\mb v_\beta-\mb v_s),\qquad k_\beta=\widehat k_\beta(\phi_2,\ldots,\phi_n).
\end{align}
\end{subequations}
This choice models a mixture of an incompressible elastic solid and $N-1$  incompressible inviscid fluids. In particular, the chemical potential tensors of the fluids are $\mb K_\beta=\mu_\beta\mb I$, where
\begin{equation}
\mu_\beta=-(\lambda+p_\beta),\qquad p_\beta:=\PART{\psi}{\phi_\beta}.
\end{equation}
The fields $p_\beta$ are called \emph{partial pressures}. The following system governing the evolution of the mixture is arrived at:
\begin{subequations}\label{eq:37++}
  \begin{align}
    &\divop\left(\PART{\psi}{\mb F}\mb F^T+\psi\mb I-\sum_{\beta}\phi_\beta\frac{\partial\psi}{\partial\phi_\beta}\right)-\gradop\lambda+\mb b=\mb 0,\label{eq:37a++}\\
    &k_\beta(\mb v_\beta-\mb v_s)=-\gradop\left(p_\beta+\lambda\right)\label{eq:37b++}.
  \end{align}
\end{subequations}
A couple of remarks are in order:
\begin{itemize}
 \item the evolution equations \eqref{eq:37} do not depend on how the total free energy $\psi$ splits into its addenda $\psi_\alpha$;
\item however, the boundary conditions in terms of stresses depend on such splitting. Thus how one chooses this splitting has relevant physical consequences;
\end{itemize}

\begin{remark} In continuum mechanics a loading environment is modeled by  prescribing boundary conditions of traction. These conditions are translated in terms of stress through the Cauchy relation (see \eqref{eq:52}), regardless what the constitutive prescription for the stress be. Once the constitutive equation for the free energy is provided, these conditions lead to a precise mathematical statement through a sequence of steps {\it free of ambiguity}. For a superposition of continua, however, there is ambiguity concerning the expression of the partial stresses because of the manifold of choices for the free energies of the individual components. As a result, the same set of boundary conditions of traction type will lead to {\it different mathematical statements} concerning boundary conditions, according to what choice has been made  for the free energies.
\end{remark}

\subsubsection{Ascribing the entire free energy to the solid constituent.} In this case $\psi_\beta=0$ for all $\beta=2,\ldots,N$. The total spatial free energy coincides with that of the solid constituent: 
\[
\psi(\mb F,\phi_\beta)=\psi_s(\mb F,\phi_\beta)=\phi_s\Psi_s(\mb F,\phi_\beta)=(\text{det}\mb F)^{-1}\Psi_s(\mb F,\phi_\beta).
\]
Then 
\[
\PART{\psi}{\mb F}\mb F^T=(\text{det}\mb F)^{-1}\PART{\Psi_s}{\mb F}\mb F^T-\psi_s.
\]
Therefore, the Cauchy stress of the solid constituent is
\[
\mb T_s=\phi_s\PART{\Psi_s}{\mb F}\mb F^T-\phi_s\lambda\mb I.
\]
Moreover, the partial pressures are given by:
\[
p_\beta=\phi_s\PART{\Psi_s}{\phi_\beta}
\]
and the Cauchy stress of the fluid constituent $\beta$ is
\[
\mb T_\beta=-\phi_\beta(p_\beta+\lambda)\mb I,\qquad .
\]
The total Cauchy stress is therefore:
\begin{equation}\label{eq:18}
\mb T=\sum_\alpha \mb T_\alpha=\phi_s\PART{\Psi_s}{\mb F}\mb F^T-\left(\sum_\beta\phi_\beta p_\beta+\lambda\right)\mb I
\end{equation}
Note that we can write
\[
\mb T_\beta=-\phi_\beta q_\beta\mb I,\qquad q_\beta=p_\beta+\lambda.
\]

\subsection{The special case of two constituents with $\psi_\beta=0$.}
\subsubsection{Spatial description}
When the number of fluid constituents reduces to 1, the volumetric constraint \eqref{eq:9} yields $\phi_f=1-\phi_s$, and hence, because of \eqref{eq:7}, $\phi_f=(\text{det}\mb F-1)/\text{det}\mb F$. Thus without loss of generality one can assume that $\Psi_s$ does not depend on $\phi_f$:\footnote{If this is not the case, one can introduce the constitutive mapping:
\[
\widetilde\Psi_s(\mb F):=\widehat\Psi_s\Big(\mb F,\frac{\text{det}\mb F-1}{\text{det}\mb F}\Big).
\]}
\[
\PART{\Psi_s}{\phi_f}=0.
\]
In this case, one has one more reason to think of the free energy as entirely pertaining to the solid constituent and assume $\psi_f=0$. In this special case, the constitutive equation \eqref{eq:18} delivering the total stress takes the form
\[
\mb T=\phi_s\PART{\Psi_s}{\mb F}\mb F^T-\lambda\mb I,
\]
and the partial pressure $p_f$ vanishes. The motion equation of the fluid constituent \eqref{eq:37b} becomes
\begin{equation}\label{eq:19}
k_f(\mb v_f-\mb v_s)=-\gradop\lambda.
\end{equation}

\section{An alternative expression of the internal power and for the dissipation inequality}
In this section we show that the power expenditure within the setting of theories based on interacting continua \cite{Bowen1980,CaldeCLZ2010JoCaTN} and within the setting of a single continuum with microstructure \cite{DudaSF2010JMPS} according to Biot's approach are equivalent, provided that the two power functionals are evaluated at states that are compatibile with the balance equations.

\subsection{The dissipation inequality for a two--component mixture.}
Again, consider a situation when there is only one fluid constituent and the free energy of the mixture (including the interaction energy) is ascribed entirely to the solid constituent, that is to say, $\psi_f=0$. In this case, the dissipation inequality for a fixed spatial region $\Omega$ is  
\begin{equation}\label{eq:20}
  \frac{\rm d}{{\rm d}t}\int_\Omega \psi_s+\int_{\partial\Omega}\psi_s\mb v_s\cdot\mb n\le\int_\Omega\left(\mb T_s\cdot\gradop \mb v_s+\mb T_f\cdot\gradop \mb v_f+\mb f_s\cdot\mb v_s+\mb f_f\cdot\mb v_f\right).
\end{equation}
For the reader's sake we recall the following facts:
\begin{itemize}
\item $\psi_s$ denotes the free energy of the solid constituent per unit spatial volume. Such quantity is related to the corresponding free energy density $\Psi_s$ per unit referential volume of the solid constituent by
\[
\psi_s=\phi_s\Psi_s,
\] 
where $\phi_s=1/{\det \mb F}$ is the volume fraction of the solid constituent.
\item $\mb T_s$ and $\mb T_f$ are the partial Cauchy stresses of the solid and of the fluid, and 
\[
\mathbf T=\mathbf T_s+\mathbf T_f
\]
is the total stress.
\item the fluid stress is $\bm T_f=-\phi_f\lambda\mathbf I$, where $\phi_f$ is the fluid volume fraction and $\lambda$ is the Lagrange multiplier associated to the incompressibility constraint.
\item $-\mb f_s$ and $-\mb f_f$ are the body forces acting on each continuum due to the interaction with the other continuum.
\item under the assumption that external body forces vanish, the principle of virtual powers entails a balance equation 
\[
-\divop\mb T_\alpha+\mb f_\alpha=\mb 0
\]
for each constituent $\alpha=1,2$.
\item invariance of the internal power with respect to superposition of translation fields entails that the interaction forces add up to null: $\mb f_f+\mb f_s=\mb 0$, so that the dissipation inequality can be rewritten as
\begin{equation}\label{eq:43}
  \frac{\rm d}{{\rm d}t}\int_\Omega \psi_s+\int_{\partial\Omega}\psi_s\mb v_s\cdot\mb n\le\int_\Omega\left(\mb T_s\cdot\gradop \mb v_s+\mb T_f\cdot\gradop \mb v_f+\mb f_f\cdot(\mb v_f-\mb v_s)\right).
\end{equation}
\end{itemize}

\subsection{The molar flux and the molar chemical potential.}\label{sec:molar-flux-molar-2}
Now, we perform a sequence of formal manipulations to rewrite the right--hand side of the dissipation inequality: 
\begin{align*}
  \mathscr W_{\rm int}(\Omega)[\mb v_s,\mb v_f]{}&=\int_\Omega\left(\mb T_s\cdot\gradop \mb v_s+\mb T_f\cdot\gradop \mb v_f+\mb f_f\cdot(\mb v_f-\mb v_s)\right)\\
  &=\int_\Omega (\mb T_s+\mb T_f)\cdot\gradop \mb v_s+\mb T_f\cdot\gradop(\mb v_f-\mb v_s)+\mb f_f\cdot(\mb v_f-\mb v_s)\label{eq:eeeee}
    \intertext{\qquad\qquad(on recalling the identity $\divop\mb T_f=\mb f_f$ and on setting $\mb T=\mb T_s+\mb T_f$)}
  &=\int_\Omega \mb T\cdot\gradop\mb v_s+\int_{\partial\Omega}\mb T_f\mb n\cdot(\mb v_f-\mb v_s)\\
  \intertext{\qquad\qquad(on defining the molar flux of fluid relative to the solid $\mb h=(\phi_f/v)(\mb v_f-\mb v_s)$, where $v$ is the molar volume)}
&=\int_\Omega \mb T\cdot\gradop\mb v_s+\int_{\partial\Omega}\frac v {\phi_f}\, \mb T_{f} \mb n\cdot\mb h\\
  \intertext{\qquad\qquad(on recalling that $\mb T_f=-\phi_f\lambda\mb I$ for the fluid constituent)}
  &=\int_\Omega \mb T\cdot\gradop\mb v_s-\int_{\partial\Omega}v \lambda\mb h\cdot\mb n
    \intertext{\qquad\qquad(on introducing the molar chemical potential $\mu=v \lambda$)}
  &=\int_\Omega \mb T\cdot\gradop\mb v_s-\int_{\partial\Omega}\mu\mb h\cdot\mb n.
\end{align*}
We have thus proved the following result.
\begin{proposition}The dissipation inequality \eqref{eq:43} is equivalent to 
\begin{equation}\label{eq:44}
  \frac{\rm d}{{\rm d}t}\int_\Omega \psi_s+\int_{\partial\Omega}\psi_s\mb v_s\cdot\mb n\le\int_\Omega \mb T\cdot\gradop\mb v_s-\int_{\partial\Omega}\mu\mb h\cdot\mb n.
\end{equation}
where, if $v$ the molar volume of the fluid molecules, then
\[
\mb h=(\phi_f/v)(\mb v_f-\mb v_s)
\]
is the \emph{molar flux of fluid relative to the solid}, and
\[
\mu=v\lambda
\]
is the \emph{molar chemical potential of the fluid}.
\end{proposition}
\begin{remark}
  The ratio
\begin{equation}\label{eq:56}
  c=\phi_f/v
\end{equation}
  is the spatial concentraion (number of moles) of fluid per unit volume.
  \end{remark}

\subsection{Recovering the dissipation inequality of the homogenized solid}
Next, consider a time-dependent spatial region $\Omega(t)$ that convects with the solid:
\[
\Omega(t)=\bm\chi_s(\mathcal P_s,t)\qquad\text{for some body part $\mathcal P_s$ of the solid}, 
\]
and coincides with $\Omega$ at a given time $t_0$:
\[
\Omega(t_0)=\Omega.
\]
It is standard to establish that, at time $t_0$,
\begin{equation}
  \Dt\int_\Omega \psi_s+\int_{\partial\Omega}\psi_s\mb v_s\cdot\mb n=\Dt\int_{\Omega(t)}\psi_s=\int_{\Omega}\phi_s\dot\Psi_s.
\end{equation}
Thus, the dissipation inequality can be given the form
\begin{equation}\label{eq:17}
  \int_\Omega \phi_s\dot\Psi_s\le\int_\Omega \mb T\cdot\gradop\mb v_s+\int_{\partial\Omega}\mu\mb h\cdot\mb n.
\end{equation}
Using \eqref{eq:17}, it is now straightforward to prove the next proposition, which yields the dissipation inequality used in \cite{DudaSF2010JMPS}, where the mixture is treated as a single solid.
\begin{proposition}Satisfaction of the the dissipation inequality \eqref{eq:44} for every spatial control volume $\Omega$ engaged by the solid component is equivalent to asking that, for each body part $\mathcal P_s$ of the solid,
\begin{equation}
  \int_{\mathcal P_s}\dot\Psi_s\le \int_{\mc P_s} \mb T_R\cdot\gradop_R\mb v_s+\int_{\partial\mathcal P_s}\mu\mb h_{R}\cdot\mb n_R,
\end{equation}
where
\begin{equation}
  \mb T_R=\mb T\mb  F^\star=(\det\mb F)\mb T\mb F^{-\sf T},\qquad \mb h_R=(\mb  F^\star)^{\sf T}\mb h=(\det \mb F)\mb F^{-1}\mb h.
\end{equation}
are, respectively, the referential stress and the referential flux, $\gradop_R$ denotes the referential gradient and $\mb n_R$ is the outward unit normal to $\mathcal P_s$. 
\end{proposition}

\subsection{Referential form of the motion equation of the fluid constituent.}
Recalling the identity 
\[
\gradop_R\lambda=\mb F^{T}\gradop \lambda.
\]
we can write the spatial form of the motion equation \eqref{eq:19} as
\[
-\gradop_R\lambda-k_f(\det\mb F)^{-1}\mb F^T\mb F\mb h_R=\mb 0.
\]
This yields the following relation for the referential flux:
\begin{equation}\label{eq:42}
\mb h_R=-k_f^{-1}(\det\mb F)(\mathbf F^T\mathbf F)^{-1}\gradop_R\lambda=-(v k_f)^{-1}(\det\mb F)(\mathbf F^T\mathbf F)^{-1}\gradop_R\mu.
\end{equation}
This relation that agrees with the constitutive equation for the flux adopted in \cite{DudaSF2010JMPS}. There is an important difference, however: according to \cite{DudaSF2010JMPS}, \eqref{eq:42} is a constitutive equation. In our case, it is the consequence of an equilibrium equation.

\section{Comparison with a theory describing perfusion in a porous matrix}
The virtual-power format for poromechanics offered in \cite{DiCarlo2017}, which has been proposed to describe perfusion and adsorption in a porous matrix, may be interpreted as being partway between the Biot-type framework and theories based on interacting continua. In accordance with Biot's point of view, the solid constituent plays a predominant role, since its reference configuration is used as the \emph{background} on which all phenomena are described. On the other hand, as in the theory of superposed continua: (i) the velocity of the fluid component, rather than its flux, appears explicitly in the power expenditure; (ii) Darcy's law emerges as a combination of a force balance and a constitutive equation. In particular, the pressure appears in the position of a force \emph{both in the virtual-work formulation and in the dissipation inequality}.

We shall limit our attention to a maimed version of that theory, whereby pores are \emph{fully saturated} and the balance principle for the interstitial fluid mass (specifically, \cite[Eq. 4]{DiCarlo2017}) does not include mass supply. In that version, the \emph{referential} internal power expended within the typical part $\mathcal P_s$ is:
\begin{equation}\label{eq:54}
  \mathscr W_{\rm int}(\mathcal P_s)[\mb v_s,\mb v_s]=\int_{\mathcal P_s}\bigl(\mb T_R\cdot\operatorname{grad}_R\mb v_s+\mb f_{f,R}\cdot(\mb v_f-\mb v_s)-p v\dot c_R\bigr).
\end{equation}
We recall that $c_R$ is the spatial concentration (density of moles) of the fluid component multiplied by the ratio $J$ between the current volume and the referential volume, and is related to the actual concentration $c$ and to the spatial volume fraction $\phi_f$ of the fluid by:  
\begin{equation}\label{eq:57}
  c_R=Jc=J v\phi_f.
\end{equation}
To compare our \eqref{eq:54} with the statement \cite[Eq. 11]{DiCarlo2017}, one should bear in mind that the latter contains both the internal and the external \emph{virtual} power expenditure. In particular, the body forces $\mb b_s$ and $\mb b_f$ in \cite[Eq. 11]{DiCarlo2017} account not only for the interaction force between the two constituents, but also the interaction of these constituent  with the exterior. Moreover, the same token leading to \eqref{eq:2}, namely, invariance of the internal power under a change of observer, entails that the internal forces acting on each constituent must sum to null. Finally, the \emph{porosity} field $\varphi$ in \cite{DiCarlo2017}, which  is interpreted as the amount of \emph{actual} volume of physical space occupied by the pores per unit volume in the \emph{reference configuration}, is the product $v c_R$ between the molar volume $v$ and the referential molar concentration $c_R$.\footnote{On denoting by $dV_{pores,cur}$, $dV_{fluid,cur}$, and by  $dN$, respectively,  the volume occupied by the pores, the volume occupied by the fluid, and the number moles of fluid contained in an infinitesimal volume $dV_{cur}$ in the \emph{current} configuration, we have, under the assumption of saturation, that $dV_{pores,cur}=dV_{fluid,cur}$. Thus, since $\frac{dN}{dV_{fluid,cur}}=v$ is the fluid molar density, and $c=\frac{dN}{dV_{cur}}$ is the current fluid concentration, we can write $\frac{dV_{pores,cur}}{dV_{cur}}= \frac{dN}{dV_{cur}}\frac{dV_{pores,cur}}{dN}=c v$. The porosity $\varphi$ in \cite{DiCarlo2017} stands for $\frac{dV_{pores,cur}}{dV_{ref}}$, thus we can write $\varphi=J\frac{dV_{pores,cur}}{dV_{cur}}$, where $J$ is the Jacobian. Thus, we can make the identification $\varphi=c_R v$.}

A key step in \cite{DiCarlo2017} is the introduction of a \emph{perfusion velocity}
\begin{equation}\label{eq:58}
  \mathbf w=\mathbf F^{-1}(\mb v_f-\mb v_s),
\end{equation}
as the material field deemed appropriate to describe the motion of the fluid relative to the solid. This field is related to the referential concentration by the \emph{referential molar balance}
\begin{equation}\label{eq:55}
  \dot c_R+\operatorname{div}_R(c_R\mb w)=0.
\end{equation}
The referential molar balance may be obtained in several ways. We illustrate a derivation at the end of the present section, starting from an analogous statement (see \eqref{eq:59} below) which we have already derived in the spatial setting.

Although \eqref{eq:55} involves rates, it is actually a \emph{holonomic constraint}. Indeed, given \emph{any} deformation map $\bm\chi_{sf}$ mapping the reference space of the solid component into itself, taken among the many which we may choose to describe the \emph{motion of the fluid component velocity} $\mb w$, the integration of \eqref{eq:55} yields that the Jacobian $J_{sf}=\det\operatorname{grad}_R\bm\chi_{sf}$ obeys:\footnote{It follow from \eqref{eq:60}  that \emph{virtual variations} of referential concentration and perfusion velocity are related by the constraint
\begin{equation*}\label{eq:3}
   \delta c_R+\operatorname{div}_R(c_R\delta\mb w)=0.
 \end{equation*}
}
 \begin{equation}\label{eq:60}
   J_{sf}c_R=\text{constant}.
 \end{equation}

With the position \eqref{eq:58} and \eqref{eq:55}, the internal-power expenditure \eqref{eq:54} takes the form of a linear functional on the pair $(\mb v_s,\mb w)$ (\emph{cf.} \cite[Eq. 12]{DiCarlo2017}): 
\begin{equation}
  \mathscr W_{\rm int}(\mathcal P_s)[\mb v_s,\mb w]=\int_{\mathcal P_s}\left(\mb T_R\cdot\gradop_R \mb v_s+\mb F^T\mb f_{f,R}\cdot\mb w+p v\operatorname{div}(c_R\mb w)\right).
\end{equation}
Now, we can perform a change of variable to write the internal power in terms of the spatial description of the fields of interest. Observe that on setting $\widetilde{\mb f}_{f}=J\mb f_{f,R}$, we have
\begin{equation}
  \int_{\mathcal P_s}\mb F^T\mb f_{f,R}\cdot\mb w=\int_{\Omega}J\mb F^T\mb f_{f,R}\cdot\mb w=\int_{\Omega}\widetilde{\mb f}_{f}\cdot(\mb v_f-\mb v_s),
\end{equation}
Moreover, we have
\begin{equation}
  \begin{aligned}
    \int_{\mathcal P_s}pv\operatorname{div}_R(c_R\mb w)&{}=\int_{\mathcal P_s}\operatorname{grad}_R(pv)\cdot c_R\mb w-\int_{\partial{\mathcal P_s}}pv c_R\mb w\cdot\mb n_R\\
    &=\int_{\Omega}J^{-1}\mb F^T\operatorname{grad}(pv)\cdot c_R\mb w-\int_{\partial{\Omega}}pv c_R\mb w\cdot\mb F^{-\star}\mb n\\
    &=\int_{\Omega}\operatorname{grad}(pv)\cdot c(\mb v_f-\mb v_s)-\int_{\partial{\Omega}}pv c(\mb v_f-\mb v_s)\cdot\mb n\\
    &=\int_{\Omega}p v\operatorname{div}(c(\mb v_f-\mb v_s))
    \\
    &=\int_{\Omega}p vc\operatorname{div}(\mb v_f-\mb v_s)+p v\nabla c\cdot(\mb v_f-\mb v_s)
    \end{aligned}
  \end{equation}
Thus, the internal power can be written as
\begin{equation}
  \mathscr W_{\rm int}(\mathcal P_s)[\mb v_s,\mb v_f]=\int_{\mathcal P_s}\left(\mb T\cdot\gradop \mb v_s-pvc\mb I\cdot\gradop(\mb v_f-\mb v_s)+(\widetilde{\mb f}_f+pv\nabla c)\cdot(\mb v_f-\mb v_s)\right).
\end{equation}
We now recover the expression \eqref{eq:61} of the internal power (in the special case $n=2$), provided we set
\begin{equation}
  \mb T_s=\mb T+pcv\mb I,\qquad \mb T_f=-pvc\mb I,\qquad \mb f_s=\widetilde{\mb f}_f+pv\nabla c,
\end{equation}
where the last term is the so-called \emph{buoyancy force}, related to spatial changes of fluid concentration (\emph{i.e.} spatial changes of porosity).

We conclude this section by checking that the referential molar balance \eqref{eq:55} may be obtained from the \emph{spatial  molar balance for the fluid components}:
\begin{equation}\label{eq:59}
  \frac{\partial c}{\partial t}+\operatorname{div}(c\mb v_f)=0,
\end{equation}
the latter being a consequence of the second equality in the chain \eqref{eq:57} and --- given that the molar volume $v$ of the (incompressible) fluid component is constant --- of the identity \eqref{eq:11}. From the relation $c_R(\mb X,t)=J(\mb X,t)c(\chi_{\rm s}(\mb X,t),t)$ and from the identity $\dot J=J\operatorname{div}\mb v_s$ we obtain
\begin{equation}
  \dot c_R=J\bigl(c\operatorname{div}{\mb v_s}+\frac{\partial c}{\partial t}+\operatorname{grad}c\cdot\mb v_s\bigr)=J\operatorname{div}\bigl(c(\mb v_s-\mb v_f)\bigr).
\end{equation}
Now, recall that for every time-dependent spatial vector field $\mb a$,
\begin{equation}
  J(\operatorname{div}\mb a)_r=\operatorname{div}_R((\mb F^*)^T\mb a_r),
\end{equation}
where $\mb a_r$, the referential description of $\mb a$, is related to $\mb a$ by $\mb a_r(\mb X,t)=\mb a(\bm\chi_s(\mb X,t),t)$. As a result we obtain 
\begin{equation}
  \dot c_R=\operatorname{div}_R(c (\mb F^*)^T(\mb v_f-\mb v_s)_r)=\operatorname{div}_R(Jc \mb F^{-1}(\mb v_f-\mb v_s))=\operatorname{div}_R(c_R\mb w),
\end{equation}
as desired.


\end{sloppypar}
\end{document}